%% file: conf-arxiv.tex
\documentclass{article}

\usepackage[utf8]{inputenc}
\usepackage[english] {babel} 

\usepackage[T1]{fontenc}      
\usepackage{amsmath,amsfonts,amssymb,amsthm}   
\usepackage{textcomp}         
\usepackage{float}
\usepackage{standalone}
\usepackage{xcomment}
\usepackage{xspace}
\usepackage{cite}
\usepackage{url}
\usepackage{stmaryrd}
\usepackage{wrapfig}
\usepackage{tikz,ifthen}
\usetikzlibrary{decorations.pathmorphing,decorations.pathreplacing,arrows,automata,patterns,
positioning}

\newcommand{\infnorm}[1]{\left\|#1\right\|_\infty}
\date{}
\input{figures/texcommands.tex}
\newcommand{\includepicture}[1]{\includegraphics{#1.pdf}}

\newcommand{\ZZ}{\mathbb{Z}}

\newcommand{\NN}{\mathbb{N}}

\newcommand{\manyoneinf}{\leq_m}

\newtheorem{theorem}{Theorem}[section] 
\newtheorem{lemma}[theorem]{Lemma} 
\newtheorem{definition}[theorem]{Definition}  
\begin{document}

\title{Hardness of Conjugacy and Factorization of multidimensional Subshifts of Finite Type}

\author{Emmanuel Jeandel\\
LIRMM, CNRS UMR 5506 - CC 477\\
161 rue Ada,  34\,095 Montpellier Cedex 5, France\\
Emmanuel.Jeandel@lirmm.fr
\and Pascal Vanier \\
Laboratoire d'informatique fondamentale de Marseille (LIF)\\
Aix-Marseille Universit\'e, CNRS\\
39 rue Joliot-Curie, 13453 Marseille Cedex 13, FRANCE\\
Pascal.Vanier@lif.univ-mrs.fr
\thanks{Both authors are partially supported
by ANR-09-BLAN-0164}
}






\maketitle

\begin{abstract}
We investigate here the hardness of conjugacy and factorization of subshifts of finite type (SFTs)
in dimension $d>1$. In particular, we prove that the factorization problem is $\Sigma^0_3$-complete
and the conjugacy problem $\Sigma^0_1$-complete in the arithmetical hierarchy. 

\noindent \emph{Keywords:}Subshift of finite type, factorization, conjugacy, arithmetical
hierarchy, computability, tilings.
\end{abstract}

\section*{Introduction}\label{S:intro}

A $d$-dimensional
Subshift of Finite Type (SFT) is the set of colorings of $\ZZ^d$ by a finite set of colors in which
a finite set of forbidden patterns never appear. One can also see them as tilings of $\ZZ^d$, and
in particular in dimension 2, they are equivalent to the usual notion of tilings introduced by
Wang~\cite{Wan1961}. SFTs are discrete dynamical systems, and as such,
their factoring and conjugacy relations are of great importance. If two dynamical systems are
conjugate, then they exhibit the same dynamics. When $X$ factors on $Y$, then the dynamic of $Y$ is
a subdynamic of the one of $X$.

Conjugacy is an equivalence relation and separates SFTs into classes. Classifying SFTs is a long
standing open problem in dimension one~\cite{Boy2008} which has been proved decidable in the
particular case of one-sided SFTs on $\NN$, see~\cite{Wil1973}. It has been known for a long time
that the problem was undecidable when given two SFTs, since it can be reduced to the emptyness
problem which is $\Sigma^0_1$-complete~\cite{Ber1964}. However, we prove here a slightly stronger
result: even by fixing the class in advance, it is still undecidable to decide whether some given
SFT belongs to it:

\begin{theorem}\label{thm:conj}
 For any fixed $X$, given $Y$ as an input, it is $\Sigma^0_1$-complete to decide if $X$ and $Y$ are
conjugate.
\end{theorem}

The more general problem of knowing if some SFT exhibits subdynamics of another has also been
studied intensively. In dimension one, it is only partly solved for the case when the
entropies of the two SFTs $X,Y$ verify $h(X)>h(Y)$, see~\cite{Boy1983}. Factor maps have also
been studied with the hope of finding universal SFTs: SFTs that can factor on any other and thus
contain the dynamics of all of them. However it has been shown that such SFTs do not exist,
see~\cite{Hoc2009,BT2010}. We prove here that it is harder to know if an SFT is a factor of another
than to know if it is conjugate to it.

\begin{theorem}\label{thm:fact}
 Given two SFTs $X,Y$ as inputs, it is $\Sigma^0_3$-complete to decide if $X$ factors onto $Y$.
\end{theorem}

An interesting open question for higher dimensions that would probably help solve the one
dimensional problem would be \emph{is conjugacy of subshifts with a recursive language decidable?}.
A positive answer to this question would solve the one dimensional case, even if the SFTs are
considered on $\NN^2$ instead of $\ZZ^2$.

The paper is organised as follows: first we give the necessary definitions and fix the notation is
section~\ref{S:preliminaries}, after what we give the proofs of theorems~\ref{thm:conj}
and~\ref{thm:fact} in sections~\ref{S:conjugacy}
and~\ref{S:factorization} respectively.

\section{Preliminary definitions}\label{S:preliminaries}

\subsection{Subshifts of finite type}

We give here some standard definitions and facts about multidimensional
subshifts, one may consult Lind~\cite{Lin2004} or Lind/Markus~\cite{LindMarkus} for more details.

Let $\Sigma$ be a finite alphabet, its elements are called \emph{symbols}, the $d$-dimensional
full shift on $\Sigma$ is the set
$\Sigma^{\ZZ^d}$ of all maps (colorings) from $\ZZ^d$ to the $\Sigma$ (the colors). For
$v\in\ZZ^d$, the shift functions
$\sigma_v:\Sigma^{\ZZ^d}\to\Sigma^{\ZZ^d}$, are defined locally by
$\sigma_v(c_x)=c_{x+v}$. The full shift equipped with the distance
$d(x,y)=2^{-\min\left\{\left\|v\right\|\middle\vert v\in\ZZ^d,x_v\neq
y_v\right\}}$ is a compact metric space on which the shift functions act
as homeomorphisms. An element of $\Sigma^{\ZZ^d}$ is called a
\emph{configuration}.

Every closed shift-invariant (invariant by application of any $\sigma_v$)
subset $X$ of $\Sigma^{\ZZ^d}$ is called a \emph{subshift}. An element of a
subshift is called a \emph{point} of this subshift.

Alternatively, subshifts can be defined with the help of forbidden patterns. A
\emph{pattern} is a function $p:P \to \Sigma$, where $P$, the \emph{support}, is a finite subset of
$\ZZ^d$. Let $\mathcal F$ be a collection of \emph{forbidden} patterns, the
subset $X_F$ of $\Sigma^{\ZZ^d}$ containing the configurations having nowhere a
pattern of $F$. More formally, $X_{\mathcal F}$ is defined by

$$X_{\mathcal F}=\left\{x\in\Sigma^{\ZZ^d}\middle\vert \forall z\in\ZZ^d,\forall
p\in F, x_{|z+P}\neq p \right\}\textrm{.}$$

In particular, a subshift is said to be a \emph{subshift of finite type} (SFT)
when the collection of forbidden patterns is finite. Usually, the patterns used
are \emph{blocks} or \emph{$r$-blocks}, that is they are defined over a finite
subset $P$ of $\ZZ^d$ of the form $B_r={\llbracket -r,r\rrbracket}^d$, $r$ is called its
\emph{radius}. We may assume that all patterns of $F$ are defined with blocks of the same
radius $r$, and say the SFT has radius $r$.


Given a subshift $X$, a pattern $p$ is said to be \emph{extensible} if there exists $x\in
X$ in which $p$ appears, $p$ is also said to be extensible to $x$. We also say that a pattern $p_1$
is extensible to a pattern $p_2$ if $p_1$ appears in $p_2$. A block or pattern is said to be
\emph{admissible} if it does not contain any forbidden pattern. Note that every extensible pattern is admissible but
that the converse is not necessarily true. As a matter of fact, for SFTs, it is undecidable (in
$\Pi^0_1$ to be precise) in general to know whether a pattern is extensible while
it is always decidable efficiently (polytime) to know if a pattern is admissible.

As we said before, SFTs are compact spaces, this gives a link between admissible and extensible: if
a pattern appears in an increasing sequence of admissible patterns, then it appears in a valid
configuration and is thus extensible. More generally, if we have an increasing sequence of
admissible pattern, then we can extract from it a sequence converging to some point of the SFT.

Note that instead of using the formalism of SFTs we could have used
the formalism of Wang tiles, in which numerous results have been proved. In particular the 
undecidability of knowing whether a SFT is empty. Since we will be using a construction based on
Wang tiles, we review their definitions.

\emph{Wang tiles} are unit squares with colored edges which may not be flipped
or rotated. A \emph{tileset} $T$ is a finite set of Wang tiles. A
\emph{coloring of the plane} is a mapping $c:\ZZ^2\rightarrow T$ assigning a
Wang tile to each point of the plane. If all adjacent tiles of a
 coloring of the plane have matching edges, it is called a
tiling.

 The set of tilings of a Wang tileset is a SFT on the alphabet
formed by the tiles. Conversely, any SFT is isomorphic to a Wang tileset. From a
recursivity point of view, one can say that SFTs and Wang tilesets are
equivalent. In this paper, we will be using both terminologies indiscriminately.

\subsection{Factorization and Conjugacy}

In the rest of the paper, we will use the notation $\Sigma_X$ for the alphabet
of the subshift $X$. 

A subshift $X\subseteq\Sigma_X^{\ZZ^2}$ \emph{factors} on a subshift $Y\subseteq\Sigma_Y^{\ZZ^2}$ if
there exists  a finite set $V=\{v_1,\dots,v_k\}\subset \ZZ^2$, the \emph{window}, and a local
map $f:\Sigma_X^{|V|}\to \Sigma_Y$, such that for any point $y\in Y$, there exists
a point $x\in X$ such that for all $z\in\ZZ^d, y_z=f(x_{z+v_1},\dots, x_{z+v_k})$.
That is to say, the global map $F:X\to Y$ associated to $f$, the
\emph{factorization} or \emph{factor map}, verifies $F(X)=Y$. Without loss of generality, we may
suppose that the window is an $r$-block, $r$ being then called the radius of $f$ and
$(2r+1)$ its diameter.

When the map $F$ is invertible and its inverse is also a factorization, the subshifts $X$ and
$Y$ are said to be \emph{conjugate}. In the rest of the paper, we will note with the same symbol the
local and global functions, the context making clear which one is being used.

The entropy of a subshift $X$ is defined as
\begin{displaymath}
  h(X) = \lim_{n\to\infty} \frac{\log E_n(X)}{n^d}
\end{displaymath}

\noindent where $E_n(X)$ is the number of extensible patterns of $X$ of support $\llbracket
0,n\rrbracket^d$ where $d$ is the dimension.
The entropy is a conjugacy invariant, that is to say, if $X$ and $Y$ are conjugate, then $h(X)=h(Y)$. It is in particular easy to see thanks to the
entropy that the full shift on $n$ symbols is not conjugate to the full shift with $n'$ symbols when $n\neq n'$.

\subsection{Arithmetical Hierarchy and computability}

We give now some background in computability theory and in particular about the arithmetical
hierarchy. More details can be found in Rogers~\cite{Rogers}.

In computability, the arithmetical hierarchy is a classification of sets according to their
logical characterization. A set $A$ is  $\Sigma^0_n$ if there exists a total
computable predicate $R$ such that $x\in A \Leftrightarrow \exists \overline{y_1},\forall
\overline{y_2}, \dots, Q \overline{y_n} R(x,\overline{y_1},\dots,\overline{y_n})$, where
$Q$ is a $\forall$ or an $\exists$ depending on the parity of $n$. A set $A$ is  $\Pi^0_n$ if
there exists a total computable predicate $R$ such that $x\in A \Leftrightarrow \forall
\overline{y_1},\exists
\overline{y_2}, \dots, Q \overline{y_n} R( x,\overline{y_1},\dots,\overline{y_n}$, where $Q$ is a
$\forall$ or an $\exists$ depending on the parity of $n$. Equivalently, a set is $\Sigma^0_n$
iff its complement is $\Pi^0_n$.

We say a set $A$ is many-one reducible to a set $B$, $A\manyoneinf B$ if there exists a computable
function $f$ such that for any $x$, $f(x)\in A \Leftrightarrow x\in B$. Given an enumeration of
Turing Machines $M_i$ with oracle $X$, the Turing \emph{jump} $X'$ of a set $X$ is the set of
integers $i$ such that $M_i$ halts on input $i$. We note $X^{(0)}=X$ and $X^{(n+1)}=(X^{(n)})'$. In
particular $0'$ is the set of halting Turing machines.

A set $A$ is $\Sigma^0_n$-hard (resp. $\Pi^0_n$) iff for any
$\Sigma^0_n$ (resp. $\Pi^0_n$) set $B$, $B\manyoneinf A$. The problem $0^{(n)}$ is
$\Sigma^0_n$-complete. Furthermore, it is $\Sigma^0_n$-complete if it is in $\Sigma^0_n$. The sets
in $\Sigma^0_1$ are also called recursively enumerable and the sets in $\Pi^0_1$ are called the
co-recursively enumerable or effectively closed sets.

\section{Conjugacy}\label{S:conjugacy}
We prove here the $\Sigma^0_1$-completeness of the conjugacy problem in dimension $d\geq 2$, even
for a fixed SFT.

\begin{theorem}\label{thm:conj:insigma1}
 Given two SFTs $X,Y$ as an input, it is $\Sigma^0_1$ to decide whether $X$ and $Y$ are conjugate.
\end{theorem}
\begin{proof}
To decide whether $X$ and $Y$ are conjugate, we have to check if there exists two local functions
 $F:\Sigma_X\to \Sigma_Y$, and $G:\Sigma_Y\to \Sigma_X$ such that $F\circ G=id_X$ and $G\circ
F=id_Y$. To do this, we define two new functions $F', G'$:
\begin{itemize}
 \item Let $r_X$ and $r_F$ be the radiuses of $X$ and $F$, then
$F':(\Sigma_X\sqcup\{\star\})^{B_{r_{F'}}}\to
\Sigma_Y\sqcup\{\star\} $ is defined with a window of radius $r_{F'}=\max(r_F,r_X)$:
\begin{displaymath}
F' = \left\{
\begin{array}{ll}
F & \textrm{if the window does not contain a forbidden pattern or a }\star \\
\star & \textrm{otherwise}\\
\end{array}\right.
\end{displaymath}

\item Let $r_Y$ and $r_G$ be the radiuses of $Y$ and $G$, then
$G':(\Sigma_Y\sqcup\{\star\})^{B_{r_{G'}}}\to
\Sigma_X\sqcup\{\star\} $ is defined with a window of radius $r_{G'}=\max(r_F,r_X)$:
\begin{displaymath}
G' = \left\{
\begin{array}{ll}
G & \textrm{if the window does not contain a forbidden pattern or a }\star \\
\star & \textrm{otherwise}\\
\end{array}\right.
\end{displaymath}
\end{itemize}

This definition may look gruesome, but is actually simple: $F'$ and $G'$ act exactly as $F$ and
$G$, except they keep track of forbidden patterns. That is, if for some pattern $p$, $F\circ G(p)$
does not contain a $\star$, then $G(p)$ does not contain a forbidden pattern. This will allow us to
prove at the same time that $F(X)\subseteq Y$ (resp. $G(Y)\subseteq X$) and that $F\circ G=id_X$
(resp. $G\circ F=id_Y$).

We want to prove that
$X$ and $Y$ are conjugate if and only if \emph{there exists $F,G$ and $k>(r_{F'}+r_{G'})$ such that
for any admissible block $b$ of size $k$ of $X$ (resp. $Y$), $F'\circ G'(b)$ (resp. $G'\circ F'(b)$)
is equal to $b$ at 0.}

We prove this by contraposition for $F'\circ G'$, the proof is identical for the other composition.
Suppose there is a configuration $x\in X$ such that $F'\circ G'(x)\neq x$, we may suppose they
differ in $0$ by shifting $x$. So for any $k$ there is an admissible block $b$ of size $k$ such that
$F'\circ G'(b)$ differs from $b$ at 0.
Conversely, if there is a sequence of blocks $b_k$ of size $k$, for $k>(r_{F'}+r_{G'})$ such
that $F'\circ G'(b_k)$ differs from $b_k$ in 0, then by compactness we can extract a sequence
converging to some configuration $x\in X$, and by construction $F'\circ G'(x)$ differs from $x$ in
0.
\qed\end{proof}

\begin{theorem}\label{thm:conj:sigma1hard}
 For any $X$, given $Y$ as an input, it is $\Sigma^0_1$-hard to decide if $X$ and $Y$ are conjugate.
\end{theorem}
\begin{proof}
We reduce the problem to $0'$, the halting problem. Given a Turing machine $M$ we construct a SFT
$Y_M$ such that $Y_M$ is conjugate to $X$ iff $M$ halts.

Let $R_M$ be Robinson's SFT~\cite{Rob1971} encoding computations of $M$: $R_M$ is empty iff $M$
halts\footnote{Robinson's SFT is in dimension 2 of course, for higher dimensions, we take the
rules that the symbol in $x\pm e_i$ equals the symbol in $x$, for $i>2$.}.

Now take the full shift on one more symbol than $X$, note it $F$. $Y_M$ is now the disjoint
union of $X$ and $R_M\times F$.

If $M$ halts, $Y_M=X$ and hence is conjugate to $X$. In the other direction, suppose $M$ does not
halt, then $R_M\times F$ has entropy strictly greater than that of $X$ and hence $Y_M$ is not
conjugate to $X$.
\qed\end{proof}

\section{Factorization}\label{S:factorization}

We start with two small examples to see why factorization is more complex than conjugacy. Here the
examples are the simplest ones possible: we fix the SFT to which we factor in a very simple way,
thus making the factor map known in advance.

\begin{theorem}
 Let $Y$ be the SFT containing exactly one configuration, a uniform configuration. Given $X$ as an input, it is $\Pi^0_1$-complete to know whether $X$ factors onto $Y$.
\end{theorem}
\begin{proof}
 In this case the factor map is forced: it has to send everything to the only symbol of $\Sigma_Y$.
And
the problem is hence equivalent to knowing whether a SFT is \emph{not} empty, which is
$\Pi^0_1$-complete.
\end{proof}

\begin{theorem}
 Let $Y$ be the empty SFT. Given $X$ as an input, it is $\Sigma^0_1$-complete to know whether $X$ factors onto $Y$.
\end{theorem}
\begin{proof}
 Here any factor map is suitable, the problem is equivalent to knowing whether $X$ is empty, which
is $\Sigma^0_1$-complete.
\end{proof}

We study now the hardness of factorization  in the general case, that is to say when two SFTs are
given as inputs and we 
want to know whether one is a factor of the other. We prove here with
theorems~\ref{thm:fact:insigma3} and~\ref{thm:fact:sigma3hard} the
$\Sigma^0_3$-completeness of the factorization problem.

\subsection{Factorization is in $\Sigma^0_3$}

\begin{theorem}\label{thm:fact:insigma3}
 Given two SFTs $X,Y$ as an input, deciding whether $X$ factors onto $Y$ is in $\Sigma^0_3$.
\end{theorem}
\begin{proof}
$X$ factors onto $Y$ iff there exists a factor map $F$, a local function, such that $F(X)=Y$. This
is the first existential quantifier.The
result follows from the two next lemmas.
\qed\end{proof}

\begin{lemma}\label{lem:fact:insigma3:pi2}
Given $F,X,Y$ as an input, deciding if $Y \subseteq F(X)$ is $\Pi^0_2$.
\end{lemma}
\begin{proof}
We prove here that the statement $Y\subseteq F(X)$, that is to say, \emph{for every point $y\in Y$,
there exists a point $x\in X$ such that $F(x)=y$}, is equivalent to the following $\Pi^0_2$
statement: \emph{for any admissible pattern $m$ of $Y$, if $m$ is extensible, then $F^{-1}(m)$
contains an admissible pattern}. This statement is $\Pi^0_2$ since checking that $m$ is \emph{not}
extensible is $\Sigma^0_1$, that is to say: there exists a radius $r$ such that all $r$-blocks
containing $m$ are not admissible.

We now prove the equivalence. Suppose that $Y\subseteq F(X)$, then any extensible pattern $m$ of $Y$
appears in a configuration $y\in Y$ which has a preimage $x\in X$. A preimage of $m$ being
extensible, it is also admissible. This proves the first direction.

Conversely, suppose all extensible patterns $m$ of $Y$ have an admissible preimage. Let $y$ be a
point of $Y$,
then we have an increasing sequence $m_i$ of extensible patterns converging to $y$. All of them have
at least one admissible preimage $m'_i$. By compactness, we can extract from this sequence a
converging subsequence, note $x$ its limit. By construction $x$ is a point of $X$ and a preimage of
$y$.

\qed\end{proof}

\begin{lemma}\label{lem:fact:insigma3:sigma}
Given $F,X,Y$ as an input, deciding if $F(X)\subseteq Y$ is $\Sigma^0_1$.
\end{lemma}
\begin{proof}
It is clear that $F(X) \subseteq Y$ if and only if $F(X)$ does not contain any configuration where
a forbidden patterns of $Y$ appears.
We now show that this is equivalent to the following $\Sigma^0_1$ statement: \emph{there exists 
a radius $r$ such
that for any admissible $r$-block $M$ of $X$, $F(M)$ does not contain any forbidden pattern
in its center.}

We prove the result by contraposition, in both directions. Suppose there is a configuration $x\in
X$ such that $F(x)$ contains a forbidden pattern. Then for any radius $r>0$, there exists an
extensible, hence admissible, pattern $M$ of size $r$ such that $F(M)$ contains a forbidden pattern
in its center.

Conversely, if for any radius $r>0$, there exists an admissible pattern $M$ of $X$ of size $r$ such
that $F(M)$ contains a forbidden pattern in its center, then by compactness, there exists a
configuration $x\in X$ such that $F(x)$ contains a forbidden pattern in its center. 
\qed\end{proof}

\subsection{Factorization is $\Sigma^0_3$-hard}

To prove the hardness, we use the base construction that we introduced in~\cite{JV2011}: we note it $T$.
This construction introduces a new way to put Turing machine computations in SFTs, in particular,
the base construction has exactly one point (up to shift) in which computations may be
encoded. We call this point \emph{configuration $\alpha$}, its schematic view is shown in
figure~\ref{fig:construction}a. The computation is encoded in the inner grid which is
sparse. Each crossing between a horizontal line and a vertical one forms a \emph{cell}. The constraints are carried along
the vertical and horizontal lines, so that we may view the encoding of the Turing machine as a tiling on the grid. For each time step, 
the tape of the Turing machine is encoded in the NW-SE
diagonals and the size of the diagonal steadily increases in size when going north-east. At each
growth of the diagonal size, it gains two cells. 

Configuration $\alpha$ is made of two layers: one producing the horizontal lines and the other the
vertical ones. The layer producing the vertical lines is shown in figure~\ref{fig:alphavertlines},
the vertical lines are the black vertical lines. The configuration producing the horizontal lines
is its exact symetric along the south-west/north-east diagonal. The key property of these layers is
that when a corner tile (the tile in the lower left corner of the first square) appears, then the point is necessarily of this form.

In the original construction, corner tiles of the horizontal and vertical layers could only be superimposed to each
other. We just change this so that instead, the corner tile of the vertical layer has
to be at position $(1,-1)$ relative to the corner tile of the horizontal one. This change does not
impact any of the properties of $T$, but simplifies a bit the proof of
lemma~\ref{lem:fact:limitedshift}.

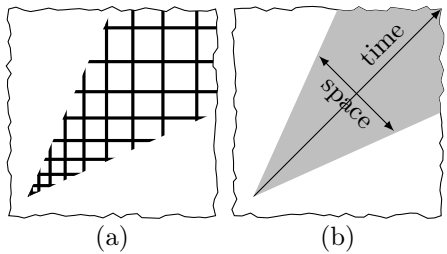
\begin{figure}[H]
 \begin{center}
 \begin{tikzpicture}[scale=0.025]
  \def\lx{10};
  \def\linewidth{1.2pt};
  \def\n{30};
  \begin{scope}[shift={(0,0)}]
  \begin{scope}
    \clip[draw,decorate,decoration={random steps, segment length=3pt, amplitude=1pt}] (-10,-10)
rectangle (100,100);
    \begin{scope}
      \clip (0,0) -- (\lx*\lx*2,\lx*\lx-\lx) -- (\lx*\lx-\lx,\lx*\lx*2) -- cycle;
      \foreach \i in {0,...,\lx}{
        \def\x{\i*\i-\i};
        \draw[line width=\linewidth] (0,\x) -- (\i*\i*2,\x);
        \draw[line width=\linewidth] (\x,0) -- (\x,\i*\i*2);
      }
    \end{scope}
  \end{scope}
  \node[below] at (45,-10) {(a)};
  \end{scope}
  \begin{scope}[shift={(120,0)}]
  \begin{scope}[even odd rule]
    \clip[draw,decorate,decoration={random steps, segment length=3pt, amplitude=1pt}] (-10,-10)
rectangle (100,100);
    \begin{scope}
      \clip (0,0) -- (\lx*\lx*2,\lx*\lx-\lx) -- (\lx*\lx-\lx,\lx*\lx*2) -- cycle;
      \fill[color=lightgray] (-10,-10) rectangle (110,110);
      \draw[-latex]  (0,0) -- (\lx*\lx,\lx*\lx) node[near end,above,sloped] {time};
     \draw[latex-latex] (\n+5,45+\n) -- (45+\n,5+\n) node[midway,below,sloped] {space};
    \end{scope}
  \end{scope}
  \node[below] at (45,-10) {(b)};
  \end{scope}

 \end{tikzpicture}
 \end{center}
 \caption{$(a)$ The skeleton of configuration $\alpha$. $(b)$ How the computation is superimposed to
$\alpha$.}
 \label{fig:construction}
\end{figure}

\begin{figure}[H]
 \begin{center}
 \scalebox{0.7}{\includepicture{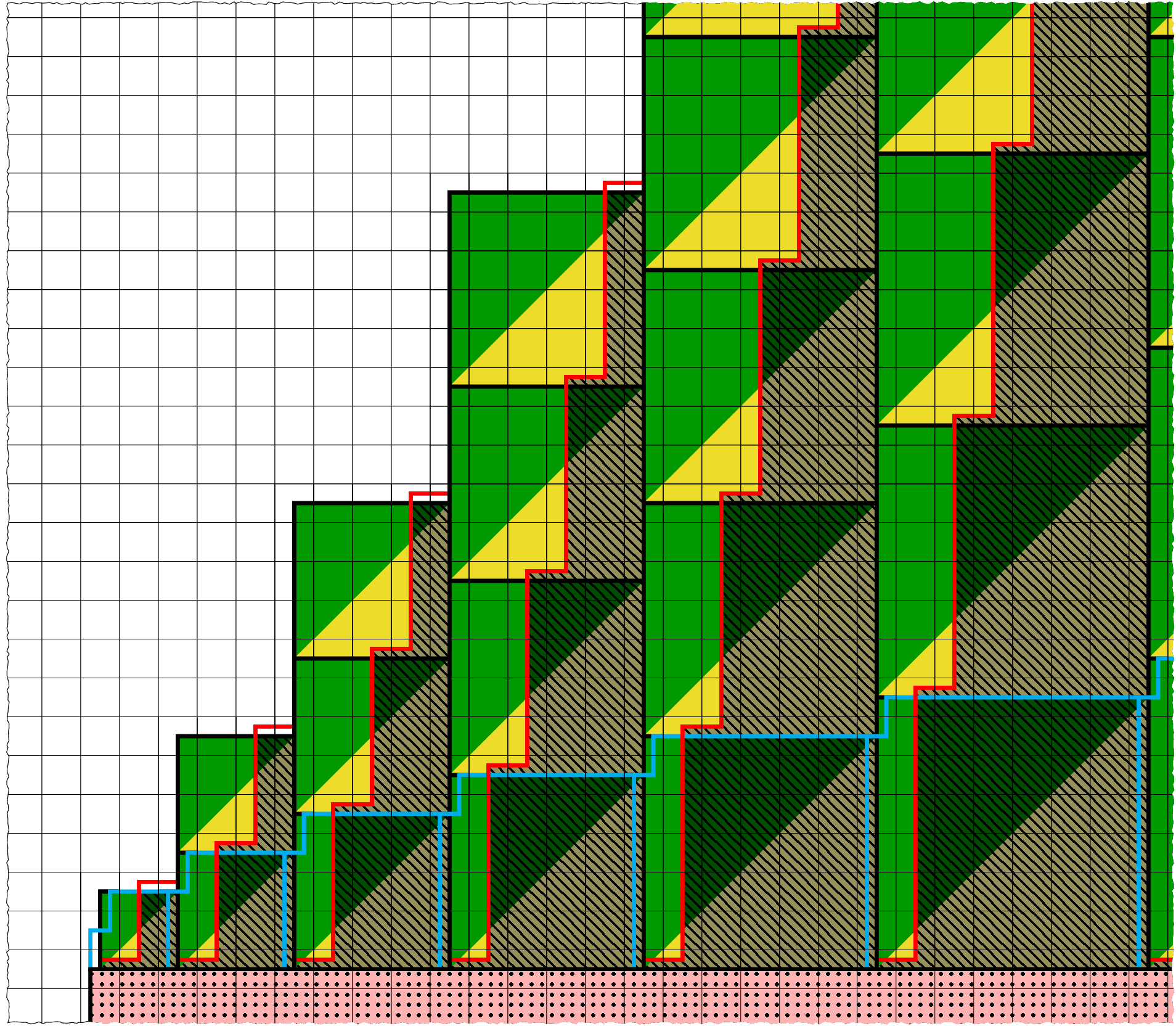}}
 \end{center}
 \caption{\label{fig:alphavertlines}Point $\alpha$, the meaningful point of $X_T$. The corner tile
may be seen on the first non all-white column: it is the lower left corner of the square.}
\end{figure}

Our reduction will use two SFTs based on this construction, both of them will be feature a different
tiling on its grid. We will say that an SFT which is basically $T$ with a tiling on
its grid as having $T$-structure.

\begin{definition}[$T$-structure]
We say an SFT $X$ has \emph{$T$-structure} it is a copy of $T$ to which we superimposed new symbols
only on the symbols representing the horizontal/vertical lines and their crossings.
\end{definition}

Note that an SFT may have $T$-structure while having no $\alpha$-configuration: for instance if you
put a computation of a Turing machine that produces an error whenever it halts.

The next lemma states a very intuitive result, that will be used later, namely that if an SFT with $T$-structure factors to another one, 
then the structure of each point is preserved by factorization. Furthermore, it shows that the factor map can only send a cell to its corresponding one, 
that is to say cell of the preimage has to be in the window of the image.

\begin{lemma}\label{lem:fact:limitedshift}
Let $X,Y$ be two SFTs with $T$-structure, such that $X$ factors onto $Y$. Let $r$ be the radius of
the factorization, then any $\alpha$-configuration of $Y$ is factored on by an
$\alpha$-configuration of $X$, and can be shifted by $v$ iff $\infnorm{v}\le r$.
\end{lemma}
\begin{proof}
By \cite[lemma 1]{JV2011}, we know that non-$\alpha$-configurations have at most one vertical line
and one horizontal line. And therefore that they have two uniform (same symbols everywhere)
quarter-planes and four uniform eighth-planes, as seen on figure~\ref{fig:planes}. The two north
east eighth-planes are not uniform in $\alpha$. thus they cannot factor on $\alpha$.

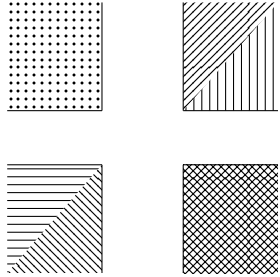
\begin{figure}
\centering
\begin{tikzpicture}[scale=0.18]
 \draw (-3,10) -- (-3,2) (-3,-2) -- (-3,-10);
 \draw (3,10) -- (3,2) (3,-2) -- (3,-10);
 \draw (-10,-2) -- (-3,-2) (3,-2) -- (10,-2);
 \draw (-10,2) -- (-3,2) (3,2) -- (10,2);
 \fill[pattern=dots] (-3,2) rectangle (-10,10);
 \fill[pattern=crosshatch] (3,-2) rectangle (10,-10);
 \fill[pattern=north east lines] (3,2) -- (10,10) -- (3,10) -- cycle;
 \fill[pattern=vertical lines] (3,2) -- (10,10) -- (10,2) -- cycle;
 \fill[pattern=north west lines] (-3,-2) -- (-10,-10) -- (-3,-10) -- cycle;
 \fill[pattern=horizontal lines] (-3,-2) -- (-10,-10) -- (-10,-2) -- cycle;
\end{tikzpicture}
\caption{\label{fig:planes} Uniform quarter- and eighth-planes in
non-$\alpha$-configurations.}
\end{figure}

It remains to prove the second part: that in the factorization process the $\alpha$-structure is at
most shifted by the radius of the factorization. We do that by contradiction, suppose that an
$\alpha$-configurations $x$ of $X$ factors on an $\alpha$-configuration $y$ of $Y$ and shifts it by
$v=(v_x,v_y)$, with $\infnorm{v}>r$. Without loss of generality we
may suppose that $v_x>r$ and $v_y>0$ and that the vertical and horizontal corner tiles of the
preimage are at positions $(0,1)$ and $(1,0)$ respectively. We are now going to show that this is
not possible.

On the horizontal layer, for all $k\in\NN^*$ there is a square with lower left corner
at $(2k^2+k,2k^2+k)$, see figure~\ref{fig:alphavertlines}. Inside this
square, there are two $(k-1)\times(k-1)$ uniform smaller squares, see
figure~\ref{fig:uniformparts}. This being also true for the vertical layer, these squares remain
uniform when they are superimposed. Now take $k$ such that $k>(\infnorm{v}+2r+1)$. By hypothesis,
there is a vertical line symbol $t$ at $z_p=(2k^2+2k+1,2k^2+k)$ on $x$, and thus at
$z_i=(2k^2+2k+1+v_x,2k^2+k+v_y)$ on $y$. We know $x_ {|z_i+B_r}$ has image $t$, and by what
precedes that $x_{|z_i+B_r}=x_{|z_i+(1,0)+B_r}$ since they are both uniform, therefore, there
should be two $t$ symbols next to eachother in $y$ at $z_i$ and $z_i+(1,0)$. This is impossible.

\begin{figure}[H]
\centering
\scalebox{0.8}{\includepicture{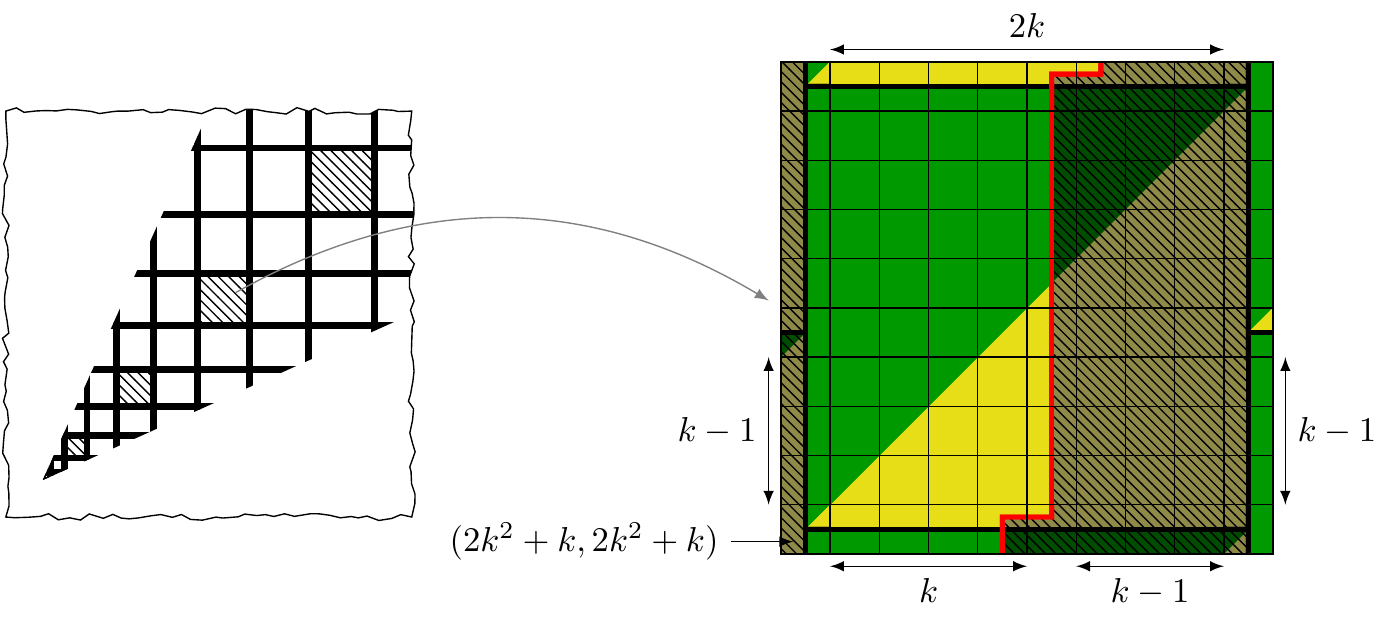}}
\caption{\label{fig:uniformparts}For every $k\in\NN^*$, the square starting at position $(2k^2+k,2k
^2+k)$ is of the form on the right on the component producing the vertical lines (and is the
symetric along the diagonal for the one producing the horizontal lines). We can see that there are
two uniform $(k-1)\times (k-1)$ squares at $(2k^2+2k+2,2k^2+k+1)$ and $(2k^2+k+1,2k^2+2k+2)$
respectively.}
\end{figure}
\qed\end{proof}

\begin{theorem}\label{thm:fact:sigma3hard}
 Given two SFTs $X,Y$ as an input, deciding whether $X$ factors onto $Y$ is $\Sigma^0_3$-hard.
\end{theorem}
\newcommand{\cofin}{\textbf{COFINITE}\xspace}
For this proof, we will reduce to the problem \cofin, which is known to be
$\Sigma^0_3$-complete, see Kozen~\cite{Kozen}. \cofin is the set of Turing
machines which run infinitely only on a finite set of inputs.

\begin{proof}

\begin{figure}[H]
\centering
 \begin{tikzpicture}[scale=0.025]
  \def\lx{10};
  \def\linewidth{1.2pt};
  \def\n{30};
  \begin{scope}[even odd rule]
    \clip[draw,decorate,decoration={random steps, segment length=3pt, amplitude=1pt}] (-10,-10)
rectangle (100,100);
    \begin{scope}
      \clip[draw] (0,0) -- (\lx*\lx*2,\lx*\lx-\lx) --
(\lx*\lx-\lx,\lx*\lx*2) -- cycle;
      \node (corner) at (0,0) {};
      \node (end1) at (\lx*\lx*2,\lx*\lx-\lx) {};
      \node (end2) at (\lx*\lx-\lx,\lx*\lx*2) {};
      \node (input1) at (\n,50+\n) {};
      \node (input2) at (50+\n,\n) {};
      \fill[color=blue!30] (\n,50+\n) -- (50+\n,\n) -- (200,\n) -- (\n,200) -- cycle;
      \draw[-latex,decorate,decoration={snake,amplitude=2mm,segment length=5mm,pre length=4mm,post
length=3mm}]
      (intersection cs:
	first line={(input1) -- (input2)},
	second line={(corner) -- (end2)}) -- (\lx*\lx,\lx*\lx); 
      \draw[latex-latex] (0,0) --  (25+\n,25+\n) node[midway,above,sloped,near end] {$d>n$};
      \draw[line width=3pt,color=black] (\n,50+\n) -- (50+\n,\n) node[midway,right]
{$n$};
      \draw[dotted, line width=1.2pt] (intersection cs:
	first line={(input1) -- (input2)},
	second line={(corner) -- (end1)})-- +(100,100);
      \draw[dotted, line width=1.2pt] (intersection cs:
	first line={(input1) -- (input2)},
	second line={(corner) -- (end2)}) -- +(100,100);
    \end{scope}
  \end{scope}
 \end{tikzpicture}
 \caption{\label{fig:zm} Computation on input $n$ in the SFT $Z$, the number of white diagonals $d$
preceeding the computation is strigtly greater than the input $n$.}
\end{figure}
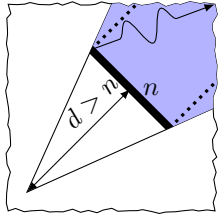

Given a Turing machine $M$, we construct two SFTs $X_M$ and $Y_M$ such
that $X_M$ factors on $Y_M$ iff the set of inputs on which $M$ does not halt is
finite. We first introduce an SFT $Z_M$ on which both will be based. It will of course have $T$
structure. Above the $T$ base, we allow the cells of the grid to be either white or blue according
to
the following rules:
\begin{itemize}
 \item All cells on a NW-SE diagonal are of the same color.
 \item A blue diagonal may follow (along direction SW-NE) a white diagonal, but not the contrary.
 \item A transition from white diagonal to blue may only appear when the grid grows.
\end{itemize}

We now allow computation on blue cells only. Only the diagonals after the growth of the grid may
contain computations. The Turing machine $M$ is launched on the input formed by the size of the
first blue line (in number of cells of course). We forbid the machine to halt

So for each $n$ on which $M$ does not halt, there is a configuration with white cells until the
first blue diagonal appears, then computation occurs inside the blue cone, see
figure~\ref{fig:zm} for a schematic view. If $M$ halts on $n$, then there is no tiling
where the first blue line codes $n$. By compactness, there is of course a configuration with only
white diagonals. If $M$ is total, then the only $\alpha$-configuration in $Z_M$ is the one with only
white diagonals.

Now from $Z_M$, we can give $X_M$ and $Y_M$:
\begin{itemize}
 \item $X_M$: Let $Z'_M$ be a copy of $Z_M$ to which we add two decorations $0$ and $1$ on the
blue cells only, and all blue cells in a configuration must have the same decoration. Now $X_M$
is $Z'_M$ to which we add a third color, red, that may only appear alone, instead of white and
blue. No computation is superimposed on red.
 \item $Y_M$ is a copy of $Z_M$ where we decorated only the horizontal corner tile with two symbols
$0$ and $1$.
\end{itemize}

We now check that $X_M$ factors onto $Y_M$ iff $M$ does not halt on a finite set of inputs:
\begin{itemize}
 \item[$\Rightarrow$] Suppose $M$ does not halt on a finite set of inputs: there exists $N$ such
that $M$ halts on every input greater than $N$. The following factor map $F$ works:
\begin{itemize}
 \item $F$ is the identity on $Z_M$. Note that the additional copy of $T$ is also sent onto $Z_M$.
 \item $F$ has a radius big enough so that when its window is centered on the corner tile, it would
cover the beginning of the computation on input $N$.
 \item An $\alpha$-configuration $x$ of $X_M$ is sent on the same $\alpha$-configuration $y$ in
$Y_M$. For the decorations, when there is a computation on $x$, the factor map can see
it and gives the same decoration to the corner tile of $y$. When there is no computation, the
factor map doesn't see a computation zone and gives decoration 0 to the corner tile. The
configuration with only white diagonals and decoration 1 of $Y_M$ is factored on by the
$\alpha$-configuration colored in red contained in $X_M$.
\end{itemize}
Note that this also works when $M$ is total.
 \item[$\Leftarrow$] Conversely, suppose $M$ does not halt on an infinite set of inputs, and that
there exists a factor map $F$ with radius $r$: lemma~\ref{lem:fact:limitedshift} states that all
$\alpha$-configurations of $Y_M$ are factored on by $\alpha$-configurations of $X_M$.
Now, there is an infinite number of $\alpha$-configurations with corner tile decorated with 0 (resp.
1) in $Y_M$, they all must be factored on by some $\alpha$-configuration of $X_M$. Still by
lemma~\ref{lem:fact:limitedshift}, the corner tile of the preimage must be in the window of the
corner tile of the image. However, there can only be a finite number of configurations in which the
symbols in this window differ. So the $\alpha$-configurations of $X_M$ factor to a finite number of
$\alpha$-configurations of $Y_M$ with one of the decorations. This is impossible.
\end{itemize}

Note that the construction of $X_M$ and $Y_M$ from the description of $M$ is computable and
uniform. The reduction is thus many-one. 
\qed\end{proof}

\bibliographystyle{alpha}
\bibliography{books,biblio}
\end{document}

%% file: figures/texcommands.tex
 \colorlet{colorCarreUpRight}{green!30!black}
 \colorlet{colorCarreUpLeft}{green!60!black}
 \colorlet{colorCarreDownRight}{yellow!50!black}
 \colorlet{colorCarreDownLeft}{yellow!90!black}
 \colorlet{colorCarreBorder}{black}
 \colorlet{colorTroisQuart}{white}
 \colorlet{colorUnQuartBas}{red!30}
 \colorlet{colorSignalCarres}{red}
 \colorlet{colorSignalQuart}{black}
 \colorlet{colorSignalBas}{cyan}

 \tikzstyle{carreUpRight}=[pattern=north west lines,pattern color=black]
 \tikzstyle{carreDownRight}=[pattern=north west lines,pattern color=black]
 \tikzstyle{carreUpLeft}=[color=colorCarreUpLeft]
 \tikzstyle{carreDownLeft}=[color=colorCarreDownLeft]
 \tikzstyle{carreBorder}=[color=colorCarreBorder]
 \tikzstyle{troisQuart}=[color=colorTroisQuart]
 \tikzstyle{unQuartBas}=[pattern=dots]
 \tikzstyle{signalCarres}=[color=colorSignalCarres]
 \tikzstyle{signalBas}=[color=colorSignalBas]
 \tikzstyle{signalQuart}=[color=colorSignalQuart]
 
 \tikzstyle{bGcarreUpRight}=[fill=colorCarreUpRight]
 \tikzstyle{bGcarreUpLeft}=[color=colorCarreUpLeft]
 \tikzstyle{bGcarreDownLeft}=[color=colorCarreDownLeft]
 \tikzstyle{bGcarreDownRight}=[color=colorCarreDownRight]
 \tikzstyle{bGcarreBorder}=[color=colorCarreBorder]
 \tikzstyle{bGtroisQuart}=[color=colorTroisQuart]
 \tikzstyle{bGunQuartBas}=[color=colorUnQuartBas]
 \tikzstyle{bGsignalCarres}=[color=colorSignalCarres]
 \tikzstyle{bGsignalBas}=[color=colorSignalBas]
 \tikzstyle{bGsignalQuart}=[color=colorSignalQuart]